\documentclass{llncs}
\usepackage[utf8]{inputenc}

\usepackage{isabelle,isabellesym}
\usepackage{amsmath}
\usepackage{amssymb}

\usepackage{url}
\usepackage{breakurl}

\usepackage[breaklinks]{hyperref}

\isabellestyle{it}

\newcommand{\CC}{{\mathcal C}}

\newcommand{\aaa}{{\tt a}}
\newcommand{\bbb}{{\tt b}}

\newcommand{\xa}{x}
\newcommand{\xb}{y}
\newcommand{\xc}{u}
\newcommand{\xd}{v}
\newcommand{\emp}{\varepsilon}
\newcommand{\pref}{{\rm pref}}

\newcommand{\abs}[1]{\left\vert #1 \right\vert}
\newcommand{\gen}[1]{\left\langle #1 \right\rangle}

\newcommand{\zero}{\ensuremath{\mathbf{0}}}
\newcommand{\one}{\ensuremath{\mathbf{1}}}
\newcommand{\ee}{\ensuremath{\mathfrak{e}}%
}
\newcommand{\ff}{\ensuremath{\mathfrak{f}}%
}

\newcommand{\isakw}[1]{\isakeyword{#1}}
\newcommand{\isaelem}[1]{#1} %for words like metis and UNIV, which is not a keyword
\newcommand{\isacomment}[1]{\hfill \textit{(* #1 *)}} %one line only! 

\usepackage{todonotes}

\usepackage[most]{tcolorbox}

\usepackage{enumitem}

\usepackage{hyperref}

\newtcolorbox{isaframe}[1][]{blanker, breakable, 
     left=3mm, right=3mm, top=1mm, bottom=1mm,
     borderline west={1pt}{0pt}{blue},
     before upper=\setlength{\parindent}{0pt},
     parbox=false, #1}

\begin{document}

\title{Binary intersection formalized}
\author{\v St\v ep\'an Holub\inst{1} and \v St\v ep\'an Starosta\inst{2}}

\institute{ 
	Dept. of Algebra, Faculty of Mathematics and Physics, Charles University, Czech Republic\\
	\email{holub@karlin.mff.cuni.cz}
	\and
	Dept. of Applied Math., Faculty of Information Technology, Czech Technical University in Prague, Czech Republic \\
	\email{stepan.starosta@fit.cvut.cz}
}
\maketitle

\begin{abstract}
	We provide a reformulation and a formalization of the classical result by Juhani Karhum\"aki characterizing intersections of two languages of the form $\{x,y\}^*\cap \{u,v\}^*$. We use the terminology of morphisms which allows to formulate the result in a shorter and more transparent way, and we formalize the result in the proof assistant Isabelle/HOL.   
\end{abstract}


\section{Introduction}
One of the classical results that deserve to be better known is the description Juhani Karhum\"aki gave in \cite{JuhaniIntersection} for the intersection of two free monoids of rank two, that is, for languages of the form $\{x,y\}^*\cap \{u,v\}^*$ where $x$ and $y$, as well as $u$ and $v$, do not commute. The purpose of this article is twofold. First, we reformulate here the result in terms of morphisms which allows an exposition that is much shorter, and hopefully also more transparent. This layer of the article is a slightly modified version of \cite{IntersectionRevisited}. Second, we complement the improved ``human'' proof with a formalization in the proof assistant Isabelle/HOL.

It is well known that an intersection of two free submonoids of a free monoid is free.  On the other hand, the intersection $\{x,y\}^*\cap \{u,v\}^*$ can have infinite rank. The Theorem 2 in \cite{JuhaniIntersection} gives two possible forms:
$\{\beta,\gamma\}^*$
and
$ (\beta_0+\beta(\gamma(1+\delta + \dots + \delta^t))^*\epsilon)^*$. The original proof spans about fifteen pages (without Preliminaries). The proof often  crucially relies on ``the way'' certain words are ``built up'' from words $x$ and $y$, and/or $u$ and $v$. This is exactly the kind of argument that is much easier to make if $x$ and $y$ ($u$ and $v$) are seen as images of a binary morphism which is demonstrated in the present article. An important feature of our reformulation is that it allowed to identify the difficult core of the proof, namely Lemma \ref{le:last_block}. Given this lemma, the rest of the proof is a fairly straightforward. We refer to \cite{IntersectionRevisited} for a more detailed comparison of the two approaches.

Our second contribution is a formalization of the result in the proof assistant Isabelle/HOL. To our knowledge, this is the first formalization of a comparable result in Combinatorics on Words. We believe that computer assisted proofs are highly desirable in our field which typically features high level of technicality. The verified formalization not only makes sure that the result is correct, but also allows to outsource tedious and uninspiring work where it belongs, namely to computers.
We try to provide a reader without any experience with this kind of research with the rough idea of what it entails. It may perhaps serve as a very modest introduction into some basic features of formalization using Isabelle/HOL.  The full working formalization is published in the repository \cite{formalcow-binaryintersection}.

\section{Preliminaries}
Words are lists of letters from a given alphabet. They form a (free) monoid with the operation of concatenation and the neutral element, the empty word, that is denoted $\emp$. If the alphabet is $\Sigma$ then the monoid of lists is typically denoted by $\Sigma^*$ using the Kleene star. There is an ambivalence in this notation. If $Q$ is a subset of a monoid $M$, then $Q^*$ denotes the submonoid generated by $Q$ in $M$, that is, more algebraically, the submonoid $\gen Q$.
However, elements of the alphabet are \emph{not} words! This is typically ignored, or at best glossed over by identification of letters with words of length one. However, in the context of the formalization, we have to keep in mind the difference. In our convention, the expression $\Sigma^*$ is equivalent to $\gen \Sigma$, which means that $\Sigma$ is not the set of letters but the set of \emph{singleton words}, that is, words of length one. In the particular case of the binary alphabet, we shall use the generating set  $A = \{\zero,\one\}$ where $\zero$ is the word $[0]$ and $\one$ the word $[1]$.

The fact that $u$ is a prefix (suffix resp.) of $v$ is denoted $u \leq_p v$ ($u \leq_s v$ resp.). If $u \leq_p v$ ($u \leq_s v$ resp.) and $u \neq v$, then $u$ is a \emph{proper} prefix (suffix resp.) of $u$. 
We shall denote the longest common prefix (suffix resp.) of $u$ and $v$ by $u\wedge_p v$ ($u\wedge_s v$ resp.). Two words are \emph{prefix-comparable}, denoted $u \bowtie v$, (\emph{suffix-comparable}, denoted $\bowtie_s$, resp.) if one of them is a prefix (suffix resp.) of the other. 
If we want to say that $u$ is a prefix (suffix resp.) of some sufficiently large power of $v$, we say that $u$ is a prefix (suffix resp.) of $v^*$. Concepts of concatenation, prefix and suffix are extended to pairs in the obvious way.

We shall use the standard notation of regular expressions to describe certain sets of words. Note that $\{u,v\}^*$ is an alternative notation for $(u+v)^*$. In regular expressions, the empty word is represented by $1$.

If $u$ is a prefix (suffix resp.) of $v$, then $u^{-1} v$ ($v u^{-1}$ resp.) denotes the unique word such that $v = uz$ ($v = zu$ resp.). The expressions $u^{-1} v$  ($v u^{-1}$ resp.) is undefined otherwise.

A pair of noncommuting words is also called a \emph{binary code}. We need the following properties of binary codes (see \cite[Lemma 3.1]{COWhandbook}).  If $u$ and $v$ do not commute, then the word $\alpha = uv \wedge_p vu$ is prefix-comparable with all words in $\{u,v\}^*$. Moreover, there are distinct letters $c_u$ and $c_v$ such that $\alpha c_u$ is prefix-comparable with each word in $u\{u,v\}^*$ and $\alpha c_v$ is prefix comparable with each word in $v\{u,v\}^*$. We shall use these facts for suffixes analogously. They directly imply a weak version of the Periodicity lemma in the following form: 
\begin{lemma}\label{pl}
	If $w$ is a common prefix (suffix resp.) of $u^*$ and $v^*$ and 
	$|u|+|v| \leq |w|$, then $u$ and $v$ commute.
\end{lemma}

A binary morphism $f$ (defined on $\{\zero,\one\}^*$) is called \emph{marked} if $\pref_1(f(\zero))\neq \pref_1(f(\one))$, where $\pref_1(u)$ denotes the first letter of $u$. For a general binary morphism $f$, its \emph{marked version} $f_m$ is the morphism defined by $f_m(u) = \alpha_f^{-1}f(u)\alpha_f$ where $\alpha_f = f(\zero\one)\wedge_p f(\one\zero)$. It is easy to see, from the facts mentioned above, that the definition of $f_m$ is correct, and that $f_m$ is marked. 

We remark that, compared to \cite{IntersectionRevisited}, we adopt a more elementary approach, and do not use the powerful technique of the free basis and the Graph lemma. While using the Graph lemma in general makes certain arguments much more comfortable, in our particular case it turns out that the exposition is only negligibly affected by this choice.

\section{Formalizing the proof using Isabelle/HOL automatic proof assistant}

Isabelle\footnote{\url{https://isabelle.in.tum.de}} is a generic proof assistant allowing a formalization of mathematical formulas and their proofs.
Isabelle was originally developed at the University of Cambridge and Technische Universität München, but now includes numerous contributions from institutions and individuals worldwide.
The most important instantiation of Isabelle to higher-order logic is Isabelle/HOL, the reader might consult for instance \cite{PaulsonNW-FAC19} for more details on Isabelle/HOL. The freely available distribution of the proof assistant also contains detailed documentation.

As mentioned in Introduction, one of the goals of this article is to provide a formalization of the presented result (and of its proof).
This is done in Isabelle/HOL.
The full formalization is available at \cite{formalcow-binaryintersection}.
In this article, we give an overview of key concepts, with comments suitable for readers not familiar with Isabelle/HOL.
If a reader is not interested in this formalization, these sections may be skipped. 

We start by introducing the formalization of the main ideas of Preliminaries.
The core building stones of Isabelle are datatypes, terms and formulae.
Our basic datatype, used for a word, is a list, which is in Isabelle equipped with many needed tools such as concatenation, denoted as multiplication.

\subsection{Words and their datatype}

To capture a word over a binary alphabet, we use a custom datatype which allows to work with all binary words.
The following code defines the datatype consisting of two values \isaelem{bin0} and \isaelem{bin1}:

\begin{isaframe}
\isacommand{datatype}\isamarkupfalse%
\ binA\ {\isacharequal}\ bin{\isadigit{0}}\ {\isacharbar}\ bin{\isadigit{1}}
\end{isaframe}

The next declarations set up abbreviations for the two words of length 1, denoted by $\zero$ and $\one$ (these are the lists of length $1$, denoted by \isaelem{[bin0]} and \isaelem{[bin1]}).

\begin{isaframe}
\isacommand{abbreviation}\isamarkupfalse%
\ bin{\isacharunderscore}word{\isacharunderscore}{\isadigit{0}}\ \ {\isacharcolon}{\isacharcolon}\ {\isachardoublequoteopen}binA\ list{\isachardoublequoteclose}\ {\isacharparenleft}{\isachardoublequoteopen}{\isasymzero}{\isachardoublequoteclose}{\isacharparenright}\ \isanewline 
\ \ \isakeyword{where} {\isachardoublequoteopen}bin{\isacharunderscore}word{\isacharunderscore}{\isadigit{0}}\ {\isasymequiv}\ {\isacharbrackleft}bin{\isadigit{0}}{\isacharbrackright}{\isachardoublequoteclose}

\isacommand{abbreviation}\isamarkupfalse%
\ bin{\isacharunderscore}word{\isacharunderscore}{\isadigit{1}}\ {\isacharcolon}{\isacharcolon}\ {\isachardoublequoteopen}binA\ list{\isachardoublequoteclose}\ {\isacharparenleft}{\isachardoublequoteopen}{\isasymone}{\isachardoublequoteclose}{\isacharparenright}\ \isanewline 
\ \ \isakeyword{where} {\isachardoublequoteopen}bin{\isacharunderscore}word{\isacharunderscore}{\isadigit{1}}\ {\isasymequiv}\ {\isacharbrackleft}bin{\isadigit{1}}{\isacharbrackright}{\isachardoublequoteclose}%
\end{isaframe}

As an example, we exhibit the claim that all lists over the constructed datatype \isaelem{binA} are generated by the two words of length $1$.
The keyword \isaelem{UNIV} stands for the set of all elements of given type (types are inferred automatically).

\begin{isaframe}
\isacommand{lemma}\isamarkupfalse%
\ A{\isacharunderscore}generates{\isacharcolon}\ {\isachardoublequoteopen}{\isasymlangle}{\isacharbraceleft}{\isasymzero}{\isacharcomma}{\isasymone}{\isacharbraceright}{\isasymrangle}\ {\isacharequal}\ UNIV{\isachardoublequoteclose}%
\isanewline
\ \ \isacommand{by}\isamarkupfalse%
\ {\isacharparenleft}metis\ A{\isacharunderscore}singletons\ basis{\isacharunderscore}gen{\isacharunderscore}monoid\ bin{\isacharunderscore}UNIV\ lists{\isacharunderscore}UNIV\ lists{\isacharunderscore}basis\ words{\isacharunderscore}univ{\isachardot}FMonoid{\isacharunderscore}axioms{\isacharparenright}%
\end{isaframe}

The proof verified by Isabelle is given on the second line.
It gives the proof method (here \isaelem{metis}) and the names of used claims (supplied in the full code).
This formalization includes most of the concepts mentioned in Preliminaries (in general, when possible, we keep the same notation in the formalization).
For instance, let us exhibit the definition of a prefix and its notation $\leq_p$:
\begin{isaframe}
\isacommand{definition}\isamarkupfalse%
\ Prefix\ {\isacharparenleft}\isakeyword{infixl}\ {\isachardoublequoteopen}{\isasymle}p{\isachardoublequoteclose}\ {\isadigit{5}}{\isadigit{0}}{\isacharparenright}\ \isakeyword{where}\ prefdef{\isacharbrackleft}simp{\isacharbrackright}{\isacharcolon}\ \ {\isachardoublequoteopen}u\ {\isasymle}p\ v\ {\isasymequiv}\ {\isasymexists}\ z{\isachardot}\ v\ {\isacharequal}\ u\ {\isasymcdot}\ z{\isachardoublequoteclose}%
\end{isaframe}

As morphisms and their marked version form an important part of used tools, we next give their formalization details, along with further Isabelle's core concepts.

\subsection{Morphisms and their marked versions}

We formalize the concept of a (general) morphism using \isakw{locale}, the Isabelle's environment used to deal with parametric theories.
In particular, a locale allows to introduce global parameters (introduced by the keyword \isakw{fixes}) and assumptions (introduced by the keyword \isakw{assumes}), thus prevents unnecessary repetition of assumptions in every lemma.
As an illustration, we exhibit a simple claim and its proof using these assumptions, called \isakw{context} in Isabelle, and delimited by keywords \isakw{begin} and \isakw{end}.

\begin{isaframe}
\isamarkuptrue%
\isacommand{locale}\isamarkupfalse%
\ morphism\ {\isacharequal}\isanewline
\ \ \isakeyword{fixes}\ f\isanewline
\ \ \isakeyword{assumes}\ morph{\isacharcolon}\ {\isachardoublequoteopen}f\ {\isacharparenleft}u\ {\isasymcdot}\ v{\isacharparenright}\ {\isacharequal}\ f\ u\ {\isasymcdot}\ f\ v{\isachardoublequoteclose}\isanewline
\isakeyword{begin}\isanewline
\ \ \isacommand{lemma}\isamarkupfalse%
\ empty{\isacharunderscore}to{\isacharunderscore}empty{\isacharcolon}\ {\isachardoublequoteopen}f\ {\isasymepsilon}\ {\isacharequal}\ {\isasymepsilon}{\isachardoublequoteclose}\isanewline
\ \ \ \ \isacommand{by}\isamarkupfalse%
\ {\isacharparenleft}metis\ morph\ self{\isacharunderscore}append{\isacharunderscore}conv{\isadigit{2}}{\isacharparenright}%
\isanewline
\isacommand{end}
\end{isaframe}

Such a lemma in the context in fact produces a claim named 
\isaelem{morphism{\isachardot}empty{\isacharunderscore}to{\isacharunderscore}empty} which is equivalent to the following lemma:

\begin{isaframe}
\isacommand{lemma}\isamarkupfalse%
\ {\isachardoublequoteopen}morphism\ f\ {\isasymLongrightarrow}\ f\ {\isasymepsilon}\ {\isacharequal}\ {\isasymepsilon}{\isachardoublequoteclose}\isanewline
\ \ \isacommand{by}\isamarkupfalse%
\ {\isacharparenleft}simp\ add{\isacharcolon}\ morphism{\isachardot}empty{\isacharunderscore}to{\isacharunderscore}empty{\isacharparenright}%
\end{isaframe}

Note that the assumption named \isaelem{morph} contains a term with two free variables, \isaelem{u} and \isaelem{v}, with no quantifiers.
As customary, such variables are understood to be universally quantified, that is, the assumption holds for all \isaelem{u} and \isaelem{v}.
Since types are inferred automatically, this assumption implies that \isaelem{u} and \isaelem{v} are lists, and \isaelem{f} is a mapping from lists to lists.

As mentioned in Introduction, we see elements of a binary code as images by a morphism. Accordingly, a binary code is formalized by extending the locale \isaelem{morphism} by an additional assumption on the images of singletons as follows.
This gives raise to a new locale \isaelem{binary{\isacharunderscore}code}:

\begin{isaframe}
\isamarkuptrue%
\isacommand{locale}\isamarkupfalse%
\ binary{\isacharunderscore}code\ {\isacharequal}\ morphism\ {\isachardoublequoteopen}f\ {\isacharcolon}{\isacharcolon}\ binA\ list\ {\isasymRightarrow}\ {\isacharprime}a\ list{\isachardoublequoteclose}\ \ \isakeyword{for}\ f\ {\isacharplus}\isanewline
\ \ \isakeyword{assumes}\ bin{\isacharunderscore}code{\isacharcolon}\ {\isachardoublequoteopen}f\ {\isasymzero}\ {\isasymcdot}\ f\ {\isasymone}\ {\isasymnoteq}\ f\ {\isasymone}\ {\isasymcdot}\ f\ {\isasymzero}{\isachardoublequoteclose}\
\end{isaframe}

The declaration
\isaelem{
{\isachardoublequoteopen}f\ {\isacharcolon}{\isacharcolon}\ binA\ list\ {\isasymRightarrow}\ {\isacharprime}a\ list{\isachardoublequoteclose}
}
specifies the datatype of the parameter \isaelem{f}.
The given datatype is a mapping from all lists over the datatype \isaelem{binA} to the lists over a generic unspecified datatype \isaelem{{\isacharprime}a}, thus setting the domain to be all binary words.

The next pointed out formalization step are the definitions of $\alpha$ and $f_m$ in the context of \isaelem{binary{\isacharunderscore}code}, i.e., for a given morphism \isaelem{f}.

\begin{isaframe}
\isacommand{definition}\isamarkupfalse%
\ \ {\isasymalpha}\ \isakeyword{where}\ {\isasymalpha}LCP{\isacharbrackleft}simp{\isacharbrackright}{\isacharcolon}\ {\isachardoublequoteopen}{\isasymalpha}\ {\isacharequal}\ \ f\ {\isacharparenleft}{\isasymzero}\ {\isasymcdot}\ {\isasymone}{\isacharparenright}\ {\isasymand}\isactrlsub p\ f\ {\isacharparenleft}{\isasymone}\ {\isasymcdot}\ {\isasymzero}{\isacharparenright}{\isachardoublequoteclose}%

\isacommand{definition}\isamarkupfalse%
\ \ f\isactrlsub m\ \isakeyword{where}\ f\isactrlsub m{\isacharunderscore}def{\isacharbrackleft}simp{\isacharbrackright}{\isacharcolon}\ {\isachardoublequoteopen}f\isactrlsub m\ {\isacharequal}\ {\isacharparenleft}{\isasymlambda}\ w{\isachardot}\ {\isacharparenleft}{\isasymalpha}{\isasyminverse}\ {\isasymcdot}\ {\isacharparenleft}f\ w{\isacharparenright}\ {\isasymcdot}\ {\isasymalpha}{\isacharparenright}{\isacharparenright}{\isachardoublequoteclose}%
\end{isaframe}

The definition of \isaelem{f\isactrlsub m} is done using a nameless function using $\lambda$-calculus conventions.

The next claim is also in the context of \isaelem{binary{\isacharunderscore}code}, giving an essential statement on $\alpha$: $\alpha$ is a prefix of $f(w)\alpha$ for every $w$.
We display the formalized proof as well; it is done by induction on the list \isaelem{w} (that is, the base case is the empty list, and the induction step proves the claim for the list \isaelem{[a]$\cdot$w} assuming that it holds for $w$).

\begin{isaframe}
\isacommand{lemma}\isamarkupfalse%
\ {\isasymalpha}w{\isasymalpha}{\isacharcolon}\ {\isachardoublequoteopen}{\isasymalpha}\ {\isasymle}p\ f\ w\ {\isasymcdot}\ {\isasymalpha}{\isachardoublequoteclose}\isanewline
\isacommand{proof}\isamarkupfalse%
{\isacharparenleft}induct\ w{\isacharparenright}\isanewline
\ \ \isacommand{case}\isamarkupfalse%
\ Nil \isacomment{case $w = \epsilon$}
\isanewline
\ \ \isacommand{then}\isamarkupfalse%
\ \isacommand{show}\isamarkupfalse%
\ {\isacharquery}case\isanewline
\ \ \ \ \isacommand{by}\isamarkupfalse%
\ simp\isanewline
\ \ \isacommand{case}\isamarkupfalse%
\ {\isacharparenleft}Cons\ a\ w{\isacharparenright} \isacomment{induction step: case $w' = aw$, $\alpha \leq_p f(w)\alpha$}
\isanewline
\ \ \isacommand{then}\isamarkupfalse%
\ \isacommand{show}\isamarkupfalse%
\ {\isacharquery}case\ \isanewline
\ \ \isacommand{proof}\isamarkupfalse%
{\isacharminus}\isanewline
\ \ \ \ \isacommand{have}\isamarkupfalse%
\ {\isachardoublequoteopen}{\isasymalpha}\ {\isasymle}p\ f\ {\isacharbrackleft}a{\isacharbrackright}\ {\isasymcdot}\ {\isasymalpha}{\isachardoublequoteclose}\ \ \isanewline
\ \ \ \ \ \ \isacommand{using}\isamarkupfalse%
\ {\isasymalpha}{\isadigit{0}}{\isasymalpha}\ {\isasymalpha}{\isadigit{1}}{\isasymalpha}\ alphabet{\isacharunderscore}or\ \isacommand{by}\isamarkupfalse%
\ metis\isanewline
\ \ \ \ \isacommand{show}\isamarkupfalse%
\ {\isacharquery}thesis\ \isanewline
\ \ \ \ \ \ \isacommand{using}\isamarkupfalse%
\ pref{\isacharunderscore}prolong{\isacharbrackleft}OF\ {\isacartoucheopen}{\isasymalpha}\ {\isasymle}p\ f\ {\isacharbrackleft}a{\isacharbrackright}\ {\isasymcdot}\ {\isasymalpha}{\isacartoucheclose}\ \ {\isacartoucheopen}{\isasymalpha}\ {\isasymle}p\ f\ w\ {\isasymcdot}\ {\isasymalpha}{\isacartoucheclose}\ {\isacharbrackright}\isanewline
\ \ \ \ \ \ \ \ hd{\isacharunderscore}word{\isacharbrackleft}of\ \ a\ w{\isacharbrackright}\isanewline
\ \ \ \ \ \ \isacommand{by}\isamarkupfalse%
\ {\isacharparenleft}metis\ append{\isacharunderscore}assoc\ morph{\isacharparenright}\ \ \isanewline
\ \ \isacommand{qed}\isamarkupfalse%
\isanewline
\isacommand{qed}\isamarkupfalse%
\end{isaframe}
This proof gives a rough idea about the level of detail contained in the formalization. Note that the induction step uses the validity of the claim for singletons (facts named $\alpha$0$\alpha$ and $\alpha$1$\alpha$) and the simple fact (called pref{\isacharunderscore}prolong) which claims that if $w \leq_p zr$ and $r \leq_p s$, then $w \leq_p zs$. The latter claim illustrates what can be considered a single step in the formalization. Note nevertheless that even this step is based on an auxiliary lemma which is proved elsewhere using even more elementary auxiliary lemmas.

\section{The result}

Let $G=\{\xa,\xb\}$ and $H=\{\xc,\xd\}$ be two binary codes, that is  $\xa \xb \neq \xb \xa$ and $\xc \xd \neq \xd \xc$.
Our aim is to describe the intersection $I = G^*\cap H^*$. The aim is achieved by a series of reformulations.

First, we shall see the languages $G^*$ and $H^*$ as ranges of the morphisms $g_0$ and $h_0$ over $A^* = \{\zero,\one\}^*$, defined by $G = \left\{ g_0(\zero), g_0(\one) \right\}$ and $H = \left\{ h_0(\zero), h_0(\one) \right\}$.
The structure of the intersection of $G^*$ and $H^*$ will follow from a stronger result: a characterization of the coincidence set of $g_0$ and $h_0$, defined by 
\[
\CC(g_0,h_0) = \{ (r,s) \in A^*\times A^* \mid g_0(r) = h_0(s) \}.
\]
Indeed, we have 
\[I  = \{g_0(r) \mid (r,s) \in \CC(g_0,h_0)\} = \{h_0(s) \mid (r,s) \in \CC(g_0,h_0)\}\,.\]

Second, instead of $\CC(g_0,h_0)$ we shall investigate 
\[
\CC(g,h) = \{ (r,s) \in A^*\times A^* \mid g(r) = h(s) \},
\]
where $g$ is the marked version of $g_0$, and $h$ is the marked version of $h_0$.
The set $\CC(g,h)$ is easier to investigate since both $g$ and $h$ are marked.
The more difficult part of the result is establishing the relationship between $\CC(g,h)$ and $\CC(g_0,h_0)$.

Assume that $I$ contains a nonempty word, that is, that there are nonempty words $r$ and $s$ such that $g_0(r) = h_0(s)$. Then both $\alpha_g$ and $\alpha_h$ are  prefixes of $g_0(r)^i = h_0(s)^i$ for a sufficiently large $i$, which implies that $\alpha_g$ and $\alpha_h$ are prefix comparable. Without loss of generality we shall suppose $\alpha_h \leq \alpha_g$. Let $\alpha = \alpha_h^{-1}\alpha_g$. Then 
\begin{equation} \label{solutions_to_marked}
g_0(r) = h_0(s) \quad \text{if and only if} \quad \alpha g(r) = h(s) \alpha.
\end{equation}  

\subsection*{Formalization: basic locales and the coincidence set}

The morphisms $g_0$ and $h_0$ are formalized as two instances of the locale \isaelem{binary{\isacharunderscore}code}, producing a new locale \isaelem{binary{\isacharunderscore}intersection{\isacharunderscore}possibly{\isacharunderscore}empty}.
This gives access to the words $\alpha_g$ and $\alpha_h$ and to marked versions of $g_0$ and $h_0$, which obtain their expected names using \isakw{notation}.
(It also gives access to the auxiliary claims of \isaelem{binary{\isacharunderscore}code} for the two morphisms.)
The assumption $\alpha_h \leq_p \alpha_g$ is then added in yet another locale.

\begin{isaframe}
\isamarkuptrue%
\isacommand{locale}\isamarkupfalse%
\ \ binary{\isacharunderscore}intersection{\isacharunderscore}possibly{\isacharunderscore}empty\ {\isacharequal}\ 
\isanewline
\ g{\isadigit{0}}{\isacharcolon}\ binary{\isacharunderscore}code\ {\isachardoublequoteopen}g\isactrlsub {\isadigit{0}}\ {\isacharcolon}{\isacharcolon}\ binA\ list\ {\isasymRightarrow}\ {\isacharprime}a\ list{\isachardoublequoteclose}\ {\isacharplus}\ h{\isadigit{0}}{\isacharcolon}\ binary{\isacharunderscore}code\ {\isachardoublequoteopen}h\isactrlsub {\isadigit{0}}\ {\isacharcolon}{\isacharcolon}\ binA\ list\ {\isasymRightarrow}\ {\isacharprime}a\ list{\isachardoublequoteclose}\ \isanewline
\ \isakeyword{for}\ g\isactrlsub {\isadigit{0}}\ h\isactrlsub {\isadigit{0}}

\isakeyword{begin}\isanewline
\isacommand{notation}\isamarkupfalse%
\ h{\isadigit{0}}{\isachardot}{\isasymalpha}\ {\isacharparenleft}{\isachardoublequoteopen}{\isasymalpha}\isactrlsub h{\isachardoublequoteclose}{\isacharparenright}
\isacomment{setting the notation $\alpha_h$ to $\alpha$ from the parent locale representing $h_0$ }
\isanewline
\isacommand{notation}\isamarkupfalse%
\ g{\isadigit{0}}{\isachardot}{\isasymalpha}\ {\isacharparenleft}{\isachardoublequoteopen}{\isasymalpha}\isactrlsub g{\isachardoublequoteclose}{\isacharparenright}%
\isanewline
\isacommand{notation}\isamarkupfalse%
\ h{\isadigit{0}}{\isachardot}f\isactrlsub m\ {\isacharparenleft}{\isachardoublequoteopen}h{\isachardoublequoteclose}{\isacharparenright}\isanewline
\isacommand{notation}\isamarkupfalse%
\ g{\isadigit{0}}{\isachardot}f\isactrlsub m\ {\isacharparenleft}{\isachardoublequoteopen}g{\isachardoublequoteclose}{\isacharparenright}\isanewline
\isakeyword{end} \isanewline

\isacommand{locale}\isamarkupfalse%
\ \ binary{\isacharunderscore}intersection\ {\isacharequal}\ \ binary{\isacharunderscore}intersection{\isacharunderscore}possibly{\isacharunderscore}empty\ {\isacharplus}\isanewline
\ \isakeyword{assumes}\ alphas{\isacharcolon}\ {\isachardoublequoteopen}{\isasymalpha}\isactrlsub h\ {\isasymle}p\ {\isasymalpha}\isactrlsub g{\isachardoublequoteclose}
\isanewline
\isakeyword{begin}
\isanewline
\isacommand{definition}\isamarkupfalse%
\ {\isasymalpha}\ \isakeyword{where}\ {\isachardoublequoteopen}{\isasymalpha}\ {\isasymequiv}\ {\isasymalpha}\isactrlsub h{\isasyminverse}\ {\isasymcdot}\ {\isasymalpha}\isactrlsub g{\isachardoublequoteclose}%
\isanewline
\isakeyword{end}
\end{isaframe}

Using the datatype used for a binary morphism (\isaelem{binA\ list\ {\isasymRightarrow}\ {\isacharprime}a\ list}), we define the coincidence set $\CC$ as follows:

\begin{isaframe}
\isamarkuptrue%
\isacommand{definition}\isamarkupfalse%
\ Coincidence{\isacharunderscore}Set\ {\isacharcolon}{\isacharcolon}\ {\isachardoublequoteopen}{\isacharparenleft}binA\ list\ {\isasymRightarrow}\ {\isacharprime}a\ list{\isacharparenright}\ {\isasymRightarrow}\ {\isacharparenleft}binA\ list\ {\isasymRightarrow}\ {\isacharprime}a\ list{\isacharparenright}\ {\isasymRightarrow}\ {\isacharparenleft}binA\ list\ {\isasymtimes}\ binA\ list{\isacharparenright}\ set{\isachardoublequoteclose}\ {\isacharparenleft}{\isachardoublequoteopen}{\isasymC}{\isacharbrackleft}{\isacharunderscore}{\isacharcomma}{\isacharunderscore}{\isacharbrackright}{\isachardoublequoteclose}{\isacharparenright}\isanewline
\ \ \isakeyword{where}\ {\isachardoublequoteopen}Coincidence{\isacharunderscore}Set\ g\ h\ {\isasymequiv}\ {\isacharbraceleft}{\isacharparenleft}r{\isacharcomma}s{\isacharparenright}{\isachardot}\ g\ r\ {\isacharequal}\ h\ s{\isacharbraceright}{\isachardoublequoteclose}%
\end{isaframe}

The crucial relation between $\CC(g_0,h_0)$ and $\CC(g,h)$ is formalized as an equivalence (denoted by \isaelem{\isasymequiv}):

\begin{isaframe}
\isacommand{lemma}\isamarkupfalse%
\ solution{\isacharunderscore}marked{\isacharunderscore}version{\isacharcolon}\ {\isachardoublequoteopen}g\isactrlsub {\isadigit{0}}\ r\ {\isacharequal}\ h\isactrlsub {\isadigit{0}}\ s\ {\isasymequiv}\ {\isasymalpha}\ {\isasymcdot}\ g\ r\ {\isacharequal}\ h\ s\ {\isasymcdot}\ {\isasymalpha}{\isachardoublequoteclose}\isanewline
\ \ \isacommand{using}\isamarkupfalse%
\ gmarked{\isachardot}f\isactrlsub m{\isacharunderscore}conjugates\ hmarked{\isachardot}f\isactrlsub m{\isacharunderscore}conjugates\ {\isasymalpha}def\isanewline
\ \ \isacommand{by}\isamarkupfalse%
\ {\isacharparenleft}smt\ append{\isacharunderscore}assoc\ append{\isacharunderscore}same{\isacharunderscore}eq\ g{\isadigit{0}}{\isachardot}f\isactrlsub m{\isacharunderscore}conjugates\ h{\isadigit{0}}{\isachardot}f\isactrlsub m{\isacharunderscore}conjugates\ same{\isacharunderscore}append{\isacharunderscore}eq{\isacharparenright}
\end{isaframe}

Again, the displayed proof references auxiliary claims that are not present in the excerpt from the whole formalization which consists of formalizing many ``obvious'' steps.

\subsection{Block structure of $\CC(g_0,h_0)$}

We call pairs $(r,s) \in \CC(g_0,h_0)$ \emph{solutions}.
$\CC(g_0,h_0)$ is a free semigroup and the elements of its minimal generating set are \emph{minimal solutions}.

The structure of $\CC(g_0,h_0)$ heavily depends on the existence of the following three pairs of words, called \emph{blocks}: We say that $(p,q)$ is the \emph{starting block} if 
$\alpha g(p) = h(q)$, and $\alpha g(p') \neq  h(q')$ for any $(p',q') < (p,q)$. Note that $\alpha_g g(p) = \alpha_h h(q)$.
We say that $(e,f)$ is the \emph{$a$-block} if $a \in \{\zero,\one\}$ is a prefix $e$, and $(e,f)$ is a minimal solution of $g$ and $h$. The $\zero$-block and $\one$-block are also called \emph{letter blocks}.
Since $g$ and $h$ are marked, the process of the construction of a solution is deterministic in the following sense. For any comparable $g(r)$ and $h(s)$ such that $g(r)\neq h(s)$, there is at most one extension of either $r$ or $s$ which keeps the images comparable. This implies the following facts:      
\begin{itemize}
	\item each block (the starting block, the $\zero$-block and the $\one$-block) is unique if it exists;
	\item any solution in $\CC(g,h)$ has a unique decomposition into letter blocks.
\end{itemize}  

Similarly, we obtain the following characterization of morphisms without the starting block.

\begin{lemma} \label{le:no_starting_block}
	If the starting block does not exist, then $\CC(g_0,h_0)$ contains at most one minimal solution.
\end{lemma}
\begin{proof}
Note that for $\alpha = \emp$, the pair $(\emp,\emp)$ is the starting block.
Therefore, the word $\alpha$ is not empty, and since $g$ and $h$ are marked and there is no starting block, the words $r$ and $s$ satisfying
\[
 \alpha g(r) = h(s) \alpha
\]
are constructed deterministically, using the mentioned procedure, letter by letter and keeping the images prefix comparable.
If such solution exists, then the first one produced by this procedure is a prefix of any other nonempty solution, and using \eqref{solutions_to_marked}, it is thus the unique minimal solution of $\CC(g_0,h_0)$.
\end{proof}

Let us further suppose that the starting block $(p,q)$ exists. Then we have the following reduction of elements of $\CC(g_0,h_0)$ to elements of $\CC(g,h)$.

\begin{lemma} \label{le:char_solutions_a}
	If the starting block $(p,q)$ exists, and $(e,f)\in \CC(g_0,h_0)$, then $(p,q)$ is a prefix of $(ep,fq)$, and $(p^{-1}ep,q^{-1}fq) \in \CC(g,h)$.
\end{lemma}
\begin{proof}
As $(p,q)$ is the starting block, and using \eqref{solutions_to_marked}, we have $\alpha g(ep) = h(fq)$.
Thus, $(p,q)$ is a prefix of $(ep,fq)$. We may write $\alpha g(p) g(p^{-1}ep) = h(q)(q^{-1}fq)$ and obtain
\[
g \left( p^{-1}ep \right) = h \left( q^{-1}fq \right). %\qedhere %llncs style has no qeds!
\]

\end{proof}

This implies that each solution has a \emph{block decomposition} by which we mean the decomposition of $(p^{-1}ep,q^{-1}fq)$ into letter blocks.

However, the structure of $\CC(g_0,h_0)$ does not necessarily mirror the simple structure of $\CC(g,h)$. Although we may be tempted to conclude that $\CC(g_0,h_0)$ consist of elements $(pep^{-1},qfq^{-1})$ where $(e,f)\in \CC(g,h)$, the problem is that $(pep^{-1},qfq^{-1})$ is ill-defined if $(p,q)$ is not a suffix of $(pe,qf)$. Instead we have the following characterization:

\begin{lemma} \label{le:char_solutions}
\[
\CC(g_0,h_0) = \left \{\left(pep^{-1},qfq^{-1}\right) \mid  \text{$(e,f) \in \CC(g,h)$ and $(p,q)\leq_s (pe,qf)$} \right\}.
\]
\end{lemma}
\begin{proof}
The inclusion $\subseteq$ is Lemma~\ref{le:char_solutions_a}.

To see the inclusion $\supseteq$, we first verify, using the properties of the starting block, that $g(e) = h(f)$ implies,  
\[
\alpha g(pep^{-1}) = h(qfq^{-1})\alpha\,.
\]
The claim now follows from \eqref{solutions_to_marked}.
\end{proof}

\subsection*{Formalization of minimal solutions and blocks}

The definition of a minimal solution (for a morphism \isaelem{g}, word \isaelem{r}, morphism \isaelem{h}, and word \isaelem{s}, in this order) is formalized in the following way, introducing a useful short notation \isaelem{g r {\isacharequal}\isactrlsub m h s}:

\begin{isaframe}
\isacommand{definition}\isamarkupfalse%
\ MinimalSolution\ {\isacharcolon}{\isacharcolon}\ {\isachardoublequoteopen}{\isacharparenleft}binA\ list\ {\isasymRightarrow}\ {\isacharprime}a\ list{\isacharparenright}\ {\isasymRightarrow}\ binA\ list\ {\isasymRightarrow}\ {\isacharparenleft}binA\ list\ {\isasymRightarrow}\ {\isacharprime}a\ list{\isacharparenright}\ {\isasymRightarrow}\ \ binA\ list\ {\isasymRightarrow}\ bool{\isachardoublequoteclose}\ {\isacharparenleft}{\isachardoublequoteopen}{\isacharparenleft}{\isacharunderscore}\ {\isacharunderscore}{\isacharparenright}\ {\isacharequal}\isactrlsub m\ {\isacharparenleft}{\isacharunderscore}\ {\isacharunderscore}{\isacharparenright}{\isachardoublequoteclose}\ {\isacharbrackleft}{\isadigit{8}}{\isadigit{0}}{\isacharcomma}{\isadigit{8}}{\isadigit{0}}{\isacharcomma}{\isadigit{8}}{\isadigit{0}}{\isacharcomma}{\isadigit{8}}{\isadigit{0}}{\isacharbrackright}\ {\isadigit{5}}{\isadigit{1}}\ {\isacharparenright}\isanewline
\ \ \isakeyword{where}\ minsoldef{\isacharcolon}\ \ {\isachardoublequoteopen}MinimalSolution\ g\ r\ h\ s\ {\isasymequiv}\ r\ {\isasymnoteq}\ {\isasymepsilon}\ {\isasymand}\ s\ {\isasymnoteq}\ {\isasymepsilon}\ {\isasymand}\ g\ r\ {\isacharequal}\ h\ s\ {\isasymand}\ {\isacharparenleft}{\isasymforall}\ r{\isacharprime}\ s{\isacharprime}{\isachardot}\ r{\isacharprime}\ {\isasymle}np\ r\ {\isasymand}\ s{\isacharprime}\ {\isasymle}p\ s\ {\isasymand}\ g\ r{\isacharprime}\ {\isacharequal}\ h\ s{\isacharprime}\ {\isasymlongrightarrow}\ r{\isacharprime}\ {\isacharequal}\ r\ {\isasymand}\ s{\isacharprime}\ {\isacharequal}\ s{\isacharparenright}{\isachardoublequoteclose}%
\isacomment{{\isasymle}np stands for nonempty prefix}
\end{isaframe}

Formalization of Lemma~\ref{le:no_starting_block}, dealing with the case of no starting block, is rewritten and proven as:

\begin{isaframe}
\isacommand{lemma}\isamarkupfalse%
\ no{\isacharunderscore}pq{\isacharunderscore}one{\isacharunderscore}minimal{\isacharcolon}\ \isanewline
\ \ \isakeyword{assumes}\ {\isachardoublequoteopen}{\isasymAnd}\ p\ q{\isachardot}\ {\isasymalpha}\ {\isasymcdot}\ g\ p\ {\isasymnoteq}\ h\ q{\isachardoublequoteclose}\isanewline
\ \ \ \ \isakeyword{and}\ {\isachardoublequoteopen}g\isactrlsub {\isadigit{0}}\ r\ {\isacharequal}\isactrlsub m\ h\isactrlsub {\isadigit{0}}\ s{\isachardoublequoteclose}\isanewline
\ \ \ \ \isakeyword{and}\ {\isachardoublequoteopen}g\isactrlsub {\isadigit{0}}\ r{\isacharprime}\ {\isacharequal}\isactrlsub m\ h\isactrlsub {\isadigit{0}}\ s{\isacharprime}{\isachardoublequoteclose}\isanewline
\ \ \isakeyword{shows}\ {\isachardoublequoteopen}{\isacharparenleft}r{\isacharcomma}s{\isacharparenright}\ {\isacharequal}\ {\isacharparenleft}r{\isacharprime}{\isacharcomma}s{\isacharprime}{\isacharparenright}{\isachardoublequoteclose}%
\end{isaframe}

The fact that there is at most one starting block is stated (and proven) in the second basic way of writing assumptions and claims in Isabelle using implications \isaelem{\isasymLongrightarrow}.

\begin{isaframe}
\isacommand{lemma}\isamarkupfalse%
\ at{\isacharunderscore}most{\isacharunderscore}one{\isacharunderscore}pq{\isacharcolon}\ {\isachardoublequoteopen}z\ {\isasymnoteq}\ {\isasymepsilon}\ {\isasymLongrightarrow}\ z\ {\isasymcdot}\ g\ r\ {\isacharequal}\ h\ s\ {\isasymLongrightarrow}\ {\isasymexists}p\ q{\isachardot}\ z\ {\isasymcdot}\ g\ p\ {\isacharequal}\ h\ q\ {\isasymand}\ {\isacharparenleft}{\isasymforall}\ r\ s{\isachardot}\ z\ {\isasymcdot}\ g\ r\ {\isacharequal}\ h\ s\ {\isasymlongrightarrow}\ {\isacharparenleft}p\ {\isasymle}p\ r\ {\isasymand}\ q\ {\isasymle}p\ s{\isacharparenright}{\isacharparenright}{\isachardoublequoteclose}%
\end{isaframe}

Note that the lemma has two assumptions, namely $z \neq \emp$ and $z g(r) = h(s)$, and the conclusion is a complicated logical formula, which itself contains an implication which is nevertheless written as {\isasymlongrightarrow}. This illustrates two levels on which the formalization operates, and which reflect the composed name ``Isabelle/HOL'' of the proof assistant we use. While the formula of the conclusion is formulated in the \emph{object logic}, namely HOL (see \cite{PaulsonNW-FAC19}), the implication \isaelem{\isasymLongrightarrow} is part of the \emph{metalogic} proper to Isabelle, called \emph{Pure}. This metalogic is best seen as an abbreviation for the natural language construction ``if \dots then''. That is, the whole claim should be read as: 
``If $z \neq \emp$, and if  $z g(r) = h(s)$, then the following formula holds \dots.''

Finally, the assumption of existence of such a starting pair is realized using a locale, with two additional assumptions called \isaelem{pq} and \isaelem{pq{\isacharunderscore}minimal}.
Lemma~\ref{le:char_solutions} is formalized within this locale.

\begin{isaframe}
\isacommand{locale}\isamarkupfalse%
\ binary{\isacharunderscore}intersection{\isacharunderscore}pq\ {\isacharequal}\ binary{\isacharunderscore}intersection{\isacharunderscore}\isanewline
\ \ \isakeyword{for}\ p\ q\ {\isacharplus}\isanewline
\ \ \isakeyword{assumes}\ \isanewline
\ \ \ \ pq{\isacharcolon}\ {\isachardoublequoteopen}{\isasymalpha}\ {\isasymcdot}\ g\ p\ {\isacharequal}\ h\ q{\isachardoublequoteclose}\isanewline
\ \ \ \ \isakeyword{and}\ \ \ pq{\isacharunderscore}minimal{\isacharcolon}\ {\isachardoublequoteopen}{\isasymalpha}\ {\isasymcdot}\ g\ p{\isacharprime}\ {\isacharequal}\ h\ q{\isacharprime}\ {\isasymLongrightarrow}\ p\ {\isasymle}p\ p{\isacharprime}\ {\isasymand}\ q\ {\isasymle}p\ q{\isacharprime}{\isachardoublequoteclose}
\isanewline
\isakeyword{begin}\isanewline
\isacommand{lemma}\isamarkupfalse%
\ char{\isacharunderscore}solutions{\isacharcolon}\ {\isachardoublequoteopen}g\isactrlsub {\isadigit{0}}\ r\ {\isacharequal}\ h\isactrlsub {\isadigit{0}}\ s\ {\isasymlongleftrightarrow}\ {\isacharparenleft}{\isasymexists}\ e\ f{\isachardot}\ g\ e\ {\isacharequal}\ h\ f\ {\isasymand}\ p\ {\isasymle}s\ p\ {\isasymcdot}\ e\ {\isasymand}\ q\ {\isasymle}s\ q\ {\isasymcdot}\ f\ {\isasymand}\ r\ {\isacharequal}\ {\isacharparenleft}p{\isasymcdot}e{\isacharparenright}{\isasymcdot}p{\isasyminverse}\ {\isasymand}\ s\ {\isacharequal}\ {\isacharparenleft}q{\isasymcdot}f{\isacharparenright}{\isasymcdot}q{\isasyminverse}{\isacharparenright}{\isachardoublequoteclose}%
\isacomment{Lemma~\ref{le:char_solutions}}
\isanewline\isakeyword{end}
\end{isaframe}

\subsection{Letter blocks as morphisms}

Since the elements of $\CC(g,h)$ decompose into letter blocks, we define morphisms $\ee$ and $\ff$ on $A^*$ where $(\ee(a),\ff(a))$ is the $a$-block. 
The morphisms are partial if some letter block does not exist. 
The characterization is finally reduced to characterizing the set $T$ satisfying the condition of Lemma~\ref{le:char_solutions}. Namely we set
\[
T = \left\{ \tau \in A^* \mid (p,q) \leq_s (p \ee (\tau),q \ff (\tau)) \right\}.
 \]

\begin{lemma}\label{le:T_prefix_code} 
If $\tau_1,\tau_1\tau_2 \in T$, then $\tau_2 \in T$.
\end{lemma}
\begin{proof}
%Let $\tau_1$ and $\tau_1\tau_2$ be two elements of $T$.
Since $p$ is a suffix of $pe(\tau_1)$, we have that $pe(\tau_2)$ is a suffix of $pe(\tau_1)e(\tau_2)$.
Since $p$ is also a suffix of $pe(\tau_1)e(\tau_2)$, we deduce that $p$ is a suffix of $pe(\tau_2)$.
Similarly, we obtain that $q$ is a suffix of $qf(\tau_2)$.
Hence $\tau_2\in T$.
\end{proof}

We also have the following simple property.
\begin{lemma}\label{le:one_block} 
	If $c^i \in T$, where $i$ is positive and  $c \in \{\zero,\one\}$, then also $c\in T$.
\end{lemma}
\begin{proof}
If $p$ is a suffix of $pc^i$, then $p$ is a suffix of $c^*$.
It implies that $p$ is a suffix of $pc$. Similarly, if $q$ is a suffix of $qc^i$, it is a suffix of $qc$.
\end{proof}

We point out three more auxiliary arguments.

\begin{lemma}\label{le:q_pref}	
	If $\zeta\one\zero^i \in T$, with $0 < i$,  then
	\begin{enumerate}[label=(\arabic*)]
		\item \label{it1:q_pref} $p$ is a proper suffix of $\ee(\zero)$. 
		\item \label{it2:q_pref} $\alpha_g g(p) \leq_s g\left(\ee(\one\zero^i)\right)$.
		\item \label{it3:q_pref} $q \leq_s \ff(\one\zero^i)$.
	\end{enumerate}
%	
%	 $\alpha_g g(p) \leq_s g\left(\ee(\one\zero^i)\right)$.
\end{lemma}	
\begin{proof}
	1. If $p$ is not a proper suffix of $\ee(\zero)$, then $p \leq_s p\,\ee(\zeta\one\zero^i)$ implies that $\ee(\zero)$ is a suffix of $p$. From $\alpha g(\ee(p))=h(\ff(q))$, $g(\ee(\zero))=h(\ff(\zero))$ and  $q \leq_s q\,\ff(\zeta\one\zero^i)$ we deduce that $(\ee(\zero),\ff(\zero))$ is a suffix of $(p,q)$, contradictiong the minimality of $(p,q)$. 

	2. Recall that $\alpha_g$ is suffix comparable with any $g(w)$, since $g(w) = \alpha_g^{-1}g_0(w)\alpha_g$. This implies that $\alpha_g g(p)$ and  $g\left(\ee(\one\zero^i)\right)$ are suffix comparable.  It is therefore enough to show that $\alpha_g g(p)$ is shorter than $g\left(\ee(\one\zero^i)\right)$.
	From \ref{it1:q_pref} we have
	\[
	\abs{g\left(\ee(\one\zero^i)\zero\right)} + \abs{g(p)} \leq \abs{g\left(\ee(\one\zero^{i-1})\right)},  
	\]
	and the claim follows from $\abs{\alpha_g} < \abs {g(\one\zero)}$.
	
	3. If $q$ is not a suffix of $\ff(\one\zero^i)$, then $\ff(\one\zero^i)$ is a proper suffix of $q$ since $q \leq_s q\,\ff(\one\zero^i)$. This contradicts \ref{it2:q_pref} in view of $\alpha g(\ee(p))=h(\ff(q))$ and $g(\ee(\one\zero^i))=h(\ff(\one\zero^i))$.
	% \qed
	\end{proof}

The most challenging part of the proof is the following lemma. It constitutes the real core of the proof.

\begin{lemma} \label{le:last_block} 
	If $\zeta c\in T$ for some $c\in\{\zero,\one\}$ and $\zeta\in\gen{\{\zero,\one\}}$, then also $c\in T$.
\end{lemma}
\begin{proof}
	Without loss of generality, let $c=\zero$. The claim follows from Lemma \ref{le:one_block} if $\tau\in \zero^*$. Let therefore $\tau=\zeta'\one\zero^i$, and
	assume  \[(p,q)\leq_s \left(pe(\zeta')\ee (\one) \ee(\zero)^i,q \ff(\zeta')\ff(\one)\ff(\zero^i)\,\right).\]
	We want to show that $(p,q)$ is a suffix of $(p \ee(\zero),q \ff(\zero))$. This is equivalent to showing that $(p,q)$ is a suffix of $(\ee(\zero)^*,\ff(\zero)^*)$. Assume the contrary.
	
	The equality $\alpha g(p)=h(q)$ and Lemma \ref{le:q_pref}~\ref{it1:q_pref} imply that $g( \ee(\zero)p^{-1})$ is a suffix of $\alpha$. 
	Since $|\alpha|<|g(\zero\one)|$, we have that $\ee(\zero)p^{-1}$ is  $\zero^m$ for some $m\geq 1$.
	
	Let $\alpha_\ff = \ff(\zero)^*\wedge_s \ff(\one)^*$, and let $c_0$ and $c_1$ be distinct letters such that $c_0\alpha_\ff\leq_s \ff(\zero)^*$ and $c_1\alpha_\ff\leq_s \ff (\one)^*$. 
	Let, moreover, $\alpha_h = h(\zero)^*\wedge _s h(\one)^*$. Then $\alpha_hh(\alpha_\ff)$ is the longest common suffix of $h(\ff(\zero)^*)$ and $h(\ff(\one)^*)$. Since $\alpha$ is a suffix of both $g(\ee(\zero)^*)$ and $g(\ee(\one)^*)$, we deduce that $\alpha$ is a suffix of $\alpha_hh(\alpha_\ff)$ and hence 
	\[
	\abs \alpha\leq \abs{h(\alpha_\ff)} + \abs{\alpha_h} \,.
	\] 
	Since $q$ is a suffix of $q\ff(\zeta')\ff(\one)\ff(\zero)^i$ and not a suffix of $\ff(\zero)^*$, we obtain that $c_1\alpha_\ff\ff(\zero)^i$ is a suffix of $q$. From $\alpha g(p)=h(q)$ and $\ee(\zero)=\zero^mp$, we now have
	$
	h(c_1\alpha_\ff\ff(\zero)^{i-1}) \leq_s \alpha g(\zero^m)^{-1}, 
	$
	which yields 
	\[
	|h(c_1\alpha_\ff)| + |g(a)| \leq \abs \alpha.
	\]
	The two inequalities above imply that $\abs{h(c_1)}+\abs{g(\zero)} \leq \abs\alpha$ and $\abs{h(c_1)}+\abs{g(\zero)} \leq \abs{\alpha_h}$. 
	Since $\alpha_h$ is a suffix of $h(c_1)^*$, $\alpha$ is a suffix of $h(\zero)^*$ and $\alpha_h$ and $\alpha$ are suffix comparable, the Periodicity lemma implies that $g(\zero)$ and $h(c_1)$ commute (see Lemma \ref{pl}). Since both $g$ and $h$ are marked, we obtain that $\ff(\zero)\in c_1^*$ which contradicts $c_0\alpha_\ff\leq_s \ff(\zero)^*$.
	% \qed
	\end{proof}
We can now have characterize the slightly surprising possibility when the intersection of two free binary monoids is infinitely generated. This happens when both letter blocks exist, but one of the singletons is not in $T$. By symmetry, we shall therefore suppose, in the following classification lemma, that $\zero \in T$ and $\one \notin T$. 
\begin{lemma}\label{le:Tform}
    Assume that both letter blocks exist, $\zero \in T$ and $\one \notin T$. 
	Then $\tau$ is a minimal element of $T$ if and only if
	$\tau = \zero$ or
	\[
	\left\{ 
	\begin{array}{l}
	\one\text{ is a prefix of } \tau, \text{ and } \\ \zero^{t+1} \text{ is its suffix, and} \\ \text{there is no other occurrence of  } \zero^{t+1}  \text{ in } \tau,
	\end{array}
		\right.	
	\]
	where $t$ is the least non negative integer such that $q\leq_s q\, \ff(\one \zero^{t+1})$.	
\end{lemma}
\begin{proof}
	From $(p,q)\leq_s (p \ee(\zero),q \ff(\zero))$, we have that $(p,q)$ is a suffix of $(\ee(\zero)^*,\ff(\zero)^*)$. Hence there exists a least non negative integer $t$ such that $(p,q)$ is a suffix of $(p\ee(\one \zero^{t+1}),\ff(\one \zero^{t+1}))$, that is, such that $ \one \zero^{t+1} \in T$.
	
	Lemma \ref{le:q_pref}, items \ref{it1:q_pref} and \ref{it3:q_pref} yield that
	\begin{equation}\label{eq:TPred_with_1}
	\zeta \one \zero^i \in T \quad \text{ if and only if } \quad i \geq t+1,
	\end{equation}
which implies that $\zeta \zero^{t+1} \in T$ for all $\zeta$.
		We may now characterize the minimal generating set of $T$.
	
	As $\zero \in T$, using Lemma~\ref{le:T_prefix_code}, we have that the only minimal generating element $\tau \in T$ starting with $\zero$ is $\tau = \zero$.
	
	Assume now that $\tau$ is a minimal generating element of $T$ starting with $\one$.
	By \eqref{eq:TPred_with_1}, $\zero^i$ is a suffix of $\tau$ with $i \geq t+1$, hence $\zero^{t+1} \leq_s \tau$.
	Let us write $\tau = \zeta \zero^{t+1} \zeta'$.
	As $\zeta \zero^{t+1} \in T$, Lemma~\ref{le:T_prefix_code} implies $\zeta' \in T$, and minimality of $\tau$ implies $\zeta' = \emp$.
	Hence, the only occurrence of $\zeta^{t+1}$ in $\tau$ is as its suffix.
	
	Assume now that $\tau$ has prefix $\one$, suffix $\zero^{t+1}$, and there is no other occurrence of $\zero^{t+1}$.
	Have $\tau = \tau_1 \tau_2$ with $\tau_1,\tau_2 \in T$ and $\tau_1$ non-empty.
	As $\one$ is a prefix of $\tau_1$, we may write $\tau_1 = \zeta \one \zero^i$, and thus by \eqref{eq:TPred_with_1} we have $i \geq t+1$, which produces an occurrence of $\zero^{t+1}$, and thus $\tau_2 = \emp$.
	Therefore, there is no decomposition of $\tau$, and it is a minimal element of $T$.
\end{proof}

\subsection*{Formalization of letter blocks, the set $T$ and the result}

We skip the formal construction of morphisms $\ee$ and $\ff$ as much more Isabelle's concepts would need to be introduced in order to explain its technical details.
We invite the reader to inspect it in the full code.

The case when only one letter block exists is treated rather implicitly in the human proof.
Nevertheless, in the formalization, we have the following explicit claim. 

\begin{isaframe}
	\isacommand{lemma}\isamarkupfalse%
	\ unique{\isacharunderscore}block{\isacharcolon}\isanewline
	\ \ \isakeyword{assumes}\ {\isachardoublequoteopen}g\ e\ {\isacharequal}\isactrlsub m\ h\ f{\isachardoublequoteclose}\ \isanewline
	\ \ \ \ \isakeyword{and}\ {\isachardoublequoteopen}{\isasymAnd}\ e{\isacharprime}\ f{\isacharprime}{\isachardot}\ g\ e{\isacharprime}\ {\isacharequal}\isactrlsub m\ h\ f{\isacharprime}{\isasymLongrightarrow}\ {\isacharparenleft}e{\isacharprime}{\isacharcomma}f{\isacharprime}{\isacharparenright}\ {\isacharequal}\ {\isacharparenleft}e{\isacharcomma}f{\isacharparenright}{\isachardoublequoteclose}\isanewline
	\ \ \ \ \isakeyword{and}\ {\isachardoublequoteopen}g\isactrlsub {\isadigit{0}}\ r\ {\isacharequal}\isactrlsub m\ h\isactrlsub {\isadigit{0}}\ s{\isachardoublequoteclose}\ \isanewline
	\ \ \isakeyword{shows}\ {\isachardoublequoteopen}{\isacharparenleft}r{\isacharcomma}s{\isacharparenright}\ {\isacharequal}\ {\isacharparenleft}p\ {\isasymcdot}\ e\ {\isasymcdot}\ p{\isasyminverse}{\isacharcomma}\ q\ {\isasymcdot}\ f\ {\isasymcdot}\ q{\isasyminverse}{\isacharparenright}{\isachardoublequoteclose}%
\end{isaframe}

The assumption of existence of both letter blocks is introduced as a locale which used further on.

\begin{isaframe}
\isacommand{locale}\isamarkupfalse%
\ binary{\isacharunderscore}intersection{\isacharunderscore}blocks\ {\isacharequal}\ binary{\isacharunderscore}intersection{\isacharunderscore}pq\ {\isacharplus}\isanewline
\ \ \isakeyword{assumes}\ \ minblock{\isadigit{0}}{\isacharcolon}\ \ {\isachardoublequoteopen}g\ {\isacharparenleft}{\isasymee}\ {\isasymzero}{\isacharparenright}\ {\isacharequal}\isactrlsub m\ h\ {\isacharparenleft}{\isasymff}\ {\isasymzero}{\isacharparenright}{\isachardoublequoteclose}\ \isakeyword{and}\isanewline
\ \ \ \ \ \ \ \ \ \ \ hdblock{\isadigit{0}}{\isacharcolon}\ \ \ {\isachardoublequoteopen}{\isasymee}\ {\isasymzero}{\isacharbang}{\isadigit{0}}\ {\isacharequal}\ bin{\isadigit{0}}{\isachardoublequoteclose}\ \isakeyword{and} \\
\mbox{\ }\isacomment{{\isasymee}\ {\isasymzero}\isacharbang{\isadigit{0}} is the first element of the list {\isasymee}\ {\isasymzero}}
\isanewline
\ \ \ \ \ \ \ \ \ \ \ minblock{\isadigit{1}}{\isacharcolon}\ \ {\isachardoublequoteopen}g\ {\isacharparenleft}{\isasymee}\ {\isasymone}{\isacharparenright}\ {\isacharequal}\isactrlsub m\ h\ {\isacharparenleft}{\isasymff}\ {\isasymone}{\isacharparenright}{\isachardoublequoteclose}\ \isakeyword{and}\isanewline
\ \ \ \ \ \ \ \ \ \ \ hdblock{\isadigit{1}}{\isacharcolon}\ \ \ {\isachardoublequoteopen}{\isasymee}\ {\isasymone}{\isacharbang}{\isadigit{0}}\ {\isacharequal}\ bin{\isadigit{1}}{\isachardoublequoteclose}
\end{isaframe}

The set $T$ is introduced as the predicate of its elements, which is more suitable for further use.

\begin{isaframe}
\isacommand{definition}\isamarkupfalse%
\ Tpred\ {\isacharcolon}{\isacharcolon}\ {\isachardoublequoteopen}binA\ list\ {\isasymRightarrow}\ bool{\isachardoublequoteclose}\ \isakeyword{where}\ {\isachardoublequoteopen}Tpred\ {\isasymtau}\ {\isasymequiv}\ p\ {\isasymle}s\ p\ {\isasymcdot}\ {\isasymee}\ {\isasymtau}\ {\isasymand}\ q\ {\isasymle}s\ q\ {\isasymcdot}\ {\isasymff}\ {\isasymtau}{\isachardoublequoteclose}\ \isanewline
\isacommand{definition}\isamarkupfalse%
\ T\ \isakeyword{where}\ {\isachardoublequoteopen}T\ {\isasymequiv}\ {\isacharbraceleft}{\isasymtau}\ {\isachardot}\ \ Tpred\ {\isasymtau}{\isacharbraceright}{\isachardoublequoteclose}%
\end{isaframe}

The relation between the solutions, the morphism $\ee$ and $\ff$, and the set $T$ (i.e., the predicate \isaelem{Tpred}), is now a consequence of a few more straightforward lemmas in Isabelle resulting in the following:
\begin{isaframe}
\isacommand{corollary}\isamarkupfalse%
\ KeyRelation{\isacharcolon}\ {\isachardoublequoteopen}{\isasymC}{\isacharbrackleft}g\isactrlsub {\isadigit{0}}{\isacharcomma}h\isactrlsub {\isadigit{0}}{\isacharbrackright}\ {\isacharequal}\ {\isacharbraceleft}{\isacharparenleft}{\isacharparenleft}p\ {\isasymcdot}\ {\isasymee}\ {\isasymtau}{\isacharparenright}\ {\isasymcdot}\ p{\isasyminverse}{\isacharcomma}{\isacharparenleft}q\ {\isasymcdot}\ {\isasymff}\ {\isasymtau}{\isacharparenright}\ {\isasymcdot}\ q{\isasyminverse}{\isacharparenright}\ \ {\isacharbar}\ {\isasymtau}{\isachardot}\ Tpred\ {\isasymtau}\ {\isacharbraceright}{\isachardoublequoteclose}%
\end{isaframe}

Formalizations of Lemmas~\ref{le:T_prefix_code} and \ref{le:last_block} are straightforward:
\begin{isaframe}
\isacommand{lemma}\isamarkupfalse%
\ T{\isacharunderscore}prefix{\isacharunderscore}code{\isacharcolon}\ \isakeyword{assumes}\ {\isachardoublequoteopen}Tpred\ {\isasymtau}\isactrlsub {\isadigit{1}}{\isachardoublequoteclose}\ \isakeyword{and}\ {\isachardoublequoteopen}Tpred\ {\isacharparenleft}{\isasymtau}\isactrlsub {\isadigit{1}}\ {\isasymcdot}\ {\isasymtau}\isactrlsub {\isadigit{2}}{\isacharparenright}{\isachardoublequoteclose}\ \isakeyword{shows}\ {\isachardoublequoteopen}Tpred\ {\isasymtau}\isactrlsub {\isadigit{2}}{\isachardoublequoteclose}%
\isacomment{Lemma~\ref{le:T_prefix_code}}
\isanewline
\isacommand{lemma}\isamarkupfalse%
\ last{\isacharunderscore}block{\isacharcolon}\ {\isachardoublequoteopen}Tpred\ {\isacharparenleft}z\ {\isasymcdot}\ {\isacharbrackleft}c{\isacharbrackright}{\isacharparenright}\ {\isasymLongrightarrow}\ Tpred\ {\isacharbrackleft}c{\isacharbrackright}{\isachardoublequoteclose}%
\isacomment{Lemma~\ref{le:last_block}}
\end{isaframe}

The human proof of Lemma~\ref{le:last_block} contains several steps which depend on some level of insight into properties of binary codes. The formalization of this proof is therefore particularly interesting and important (and demanding). The main proof is preceded by a dedicated locale that contains forty three claims, including the claims of Lemma~\ref{le:q_pref}. In a sense, therefore, the proof of the lemma is fragmented into forty three smaller steps. It should be made clear, however, that the fragmentation is to a great extent a matter of taste, since a single proof can be often quite naturally divided into several lemmas, or vice versa. Moreover, fourteen lemmas out of the forty three are of purely preparatory nature, allowing to use other claims formulated for prefixes in a reversed way for suffixes. This is something which in the given human proof is done by a simple appeal to a ``mirrored situation'', an insight that is hardly possible to formalize in a uniform way.   

We do not list the formalized equivalents of Lemmas~\ref{le:q_pref} and \ref{le:one_block} as they are split in the code into several lemmas.

The characterization of the set $T$ is concluded in the two following locales, the first, called \isaelem{binary{\isacharunderscore}intersection{\isacharunderscore}blocks{\isacharunderscore}trivial}, is for the case $\zero, \one \in T$, the second, named \isaelem{binary{\isacharunderscore}intersection{\isacharunderscore}blocks{\isacharunderscore}nontrivial}, for the case $\zero \in T, \one \not \in T$.
The term \isaelem{\isactrlbold B\ T} stands for a basis of $T$, i.e., the set of its minimal elements, and the term \isaelem{Suc\ t} represents $t+1$.

\begin{isaframe}
\isacommand{locale}\isamarkupfalse%
\ binary{\isacharunderscore}intersection{\isacharunderscore}blocks{\isacharunderscore}trivial\ {\isacharequal}\ binary{\isacharunderscore}intersection{\isacharunderscore}blocks\ {\isacharplus}\isanewline
\ \ \isakeyword{assumes}\ \isanewline
\ \ \ \ easy{\isacharunderscore}block{\isadigit{0}}{\isacharcolon}\ {\isachardoublequoteopen}{\isacharparenleft}p\ {\isasymle}s\ p\ {\isasymcdot}\ {\isasymee}\ {\isasymzero}\ {\isasymand}\ q\ {\isasymle}s\ q\ {\isasymcdot}\ {\isasymff}\ {\isasymzero}{\isacharparenright}{\isachardoublequoteclose}\isanewline
\ \ \ \ \isakeyword{and}\ easy{\isacharunderscore}block{\isadigit{1}}{\isacharcolon}\ {\isachardoublequoteopen}{\isacharparenleft}p\ {\isasymle}s\ p\ {\isasymcdot}\ {\isasymee}\ {\isasymone}\ {\isasymand}\ q\ {\isasymle}s\ q\ {\isasymcdot}\ {\isasymff}\ {\isasymone}{\isacharparenright}{\isachardoublequoteclose}\isanewline
\isakeyword{begin}\isanewline
\isanewline
\isacommand{theorem}\isamarkupfalse%
\ {\isachardoublequoteopen}Tpred\ {\isasymtau}{\isachardoublequoteclose} \isacomment{i.e., $T = \gen{\{\zero,\one\}}$}
\isanewline
\isanewline
\isacommand{end}

\isanewline

\isacommand{locale}\isamarkupfalse%
\ binary{\isacharunderscore}intersection{\isacharunderscore}blocks{\isacharunderscore}nontrivial\ {\isacharequal}\ binary{\isacharunderscore}intersection{\isacharunderscore}blocks\isanewline
\ \ \isakeyword{for}\ t\ {\isacharplus}\isanewline
\ \ \isakeyword{assumes}\ \isanewline
\ \ \ \ \ \ \ \ easy{\isacharunderscore}block{\isacharcolon}\ {\isachardoublequoteopen}{\isacharparenleft}p\ {\isasymle}s\ p\ {\isasymcdot}\ {\isasymee}\ {\isasymzero}\ {\isasymand}\ q\ {\isasymle}s\ q\ {\isasymcdot}\ {\isasymff}\ {\isasymzero}{\isacharparenright}{\isachardoublequoteclose}\ \isanewline
\ \ \ \ \isakeyword{and}\ t{\isacharunderscore}block{\isacharcolon}\ \ {\isachardoublequoteopen}{\isasymnot}\ q\ {\isasymle}s\ q\ {\isasymcdot}\ {\isasymff}\ {\isasymone}\ {\isasymcdot}\ {\isasymff}\ {\isasymzero}{\isacharcircum}t{\isachardoublequoteclose}\isanewline
\ \ \ \ \isakeyword{and}\ t{\isacharunderscore}block{\isacharunderscore}suc{\isacharcolon}\ \ {\isachardoublequoteopen}q\ {\isasymle}s\ q\ {\isasymcdot}\ {\isasymff}\ {\isasymone}\ {\isasymcdot}\ {\isasymff}\ {\isasymzero}{\isacharcircum}Suc\ t{\isachardoublequoteclose}
\isanewline\isacommand{begin}\isanewline
\isanewline
\isacommand{corollary}\isamarkupfalse%
\ Tbasis{\isacharcolon}\ {\isachardoublequoteopen}\isactrlbold B\ T\ {\isacharequal}\ {\isacharbraceleft}{\isasymtau}{\isachardot}\ {\isasymtau}\ {\isacharequal}\ {\isasymzero}\ {\isasymor}\ {\isacharparenleft}{\isasymone}\ {\isasymle}p\ {\isasymtau}\ {\isasymand}\ {\isasymzero}{\isacharcircum}Suc\ t\ {\isasymle}s\ {\isasymtau}\ {\isasymand}\ {\isasymnot}\ {\isasymzero}{\isacharcircum}Suc\ t\ {\isasymle}f\ butlast\ {\isasymtau}{\isacharparenright}{\isacharbraceright}{\isachardoublequoteclose}%
\isacomment{Lemma~\ref{le:Tform}}
\isanewline
\isanewline\isacommand{end}
\end{isaframe}

Let us explain the notation in the claim \isaelem{Tbasis}: \isaelem{{\isasymle}f} stands for ``is factor of'' and the function butlast returns the list without its last element.

\section{Summary of the proof}
Returning from the coincidence set back to the intersection properly speaking, the main claim (Theorem 2) of  \cite{JuhaniIntersection} is that if $\{x,y\}$ and $\{u,v\}$ are binary codes, then the intersection $I = \{x,y\}^* \cap \{u,v\}^*$
has one of the following forms:
\begin{align}
\tag{$*$} I &= \{\beta,\gamma\}^* \\
\tag{$**$} I &=\left(\beta_0 + \beta(\gamma(1+\delta+ \cdots + \delta^t))^*\epsilon\right)^*
\end{align}
Let us summarize our proof and show that it agrees with the formulation from \cite{JuhaniIntersection}.
Recall that, by definition, we have $\{x,y\}= \{g_0(\zero),g_0(\one)\}$ and $\{u,v\}= \{h_0(\zero),h_0(\one)\}$.

{\bf 0.} If $I = \{\emp\}$, then the claim holds for $\beta = \gamma = \emp$. 

{\bf 1.} Let therefore $I$ contain a nonempty word. That is, $\CC(g_0,h_0)$ contains at least one minimal solution. Then $\alpha_g$ and $\alpha_h$ are prefix comparable. By symmetry, we assume $\alpha_h \leq \alpha_g$ and $\alpha = \alpha_h^{-1}\alpha_g$  
is well defined.

{\bf 1.1.} If there is no starting block, then the construction of a solution is deterministic, hence $\CC(g_0,h_0)$ contains a unique minimal solution $(r,s)$. Then $I = \{\beta,\gamma\}^*$ with $\beta = g_0(r) = h_0(s)$ and $\gamma = \emp$.

{\bf 1.2.} Let now the starting block exist, i.e., there exist $(p,q)$ such that $\alpha g(p) = h(q)$. Then each solution $(r,s)$ has a block decomposition $\tau$. We define non erasing morphisms $\ee,\ff: \{\zero,\one\}^*\to \{\zero,\one\}^*$ such that, for a solution $(r,s)$ with the block decomposition $\tau$, we have $g(\ee(\tau)) = h (\ff(\tau))$. Let $T$ be the set of block decompositions of all solutions. That is, let
   \[ \CC(g_0,h_0) = \{ (p\ee(\tau)p^{-1},q\ff(\tau)q^{-1}) \mid \tau \in T\}\,.\]
   
Note that at this moment we do not guarantee that $g(\ee(c))= h(\ff(c))$, $c \in \{\zero,\one\}$, that is, $(\ee(c),\ff(c))$ need not be defined. Because of the existence of at least one minimal solution, we may however assume, by symmetry, that $\zeta\zero \in T$ for some $\zeta$. Then $\zero \in T$ by Lemma \ref{le:last_block}, in particular $g(\ee(\zero))= h(\ff(\zero))$.  

{\bf 1.2.1.} If $(\ee(\one),\ff(\one))$ is not a letter block, then $T = \zero^*$, and $I = \{\beta,\gamma\}^*$ with 
\begin{align*}
\beta &= g_0\left(p\,\ee(\zero)\,p^{-1}\right) = h_0\left(q\,\ff(\zero)\,q^{-1}\right), & \gamma &= \emp.
\end{align*}

{\bf 1.2.2.} Suppose that $(\ee(\one),\ff(\one))$ is a letter block.

{\bf 1.2.2.1.} If $\one \in T$, then  $T = \{\zero,\one\}^*$, and
$I = \{\beta,\gamma\}^*$ with 
\begin{align*}
\beta &= g_0\left(p\,\ee(\zero)\,p^{-1}\right) = h_0\left(q\,\ff(\zero)\,q^{-1}\right), & \gamma &= g_0\left(p\,\ee(\one)\,p^{-1}\right) = h_0\left(q\,\ff(\one)\,q^{-1}\right).
\end{align*}

{\bf 1.2.2.2.} If $\one \notin T$, then by Lemma \ref{le:Tform}, there is a non negative integer $t$ such that
 \[
 T = \left( \zero+ \left(\one + \one \zero + \cdots + \one \zero^t\right)^*\one \zero^{{t+1}}\right)^*\,.
 \]
Using Lemma \ref{le:q_pref}~\ref{it2:q_pref}, we now have 
\[
I =\left(\beta_0 + \beta(\gamma(1+\delta+ \cdots + \delta^t))^*\epsilon\right)^*
\]
where 
\begin{align*}
\beta_0 &= g_0\left(p\,\ee(\zero)\,p^{-1}\right) = h_0\left(q\,\ff(\zero)\,q^{-1}\right) \\
\beta &= \alpha_g g(p)  = \alpha_h h(q) \\
\gamma &=  g\left(\ee(\one)\right) = h\left(\ff(\one)\right)\\
\delta &= g\left(\ee(\zero)\right) = h\left(\ff(\zero)\right)\\
\epsilon &= g\left(\ee(\one\zero^{t+1})\,p^{-1}\right)\alpha_g^{-1}   = h\left(\ff(\one\zero^{t+1})\,q^{-1}\right)\alpha_h^{-1}  \,.
\end{align*}
This last case, in which the intersection is infinitely generated, is further specified in \cite[Theorem 3]{JuhaniIntersection}. The generating set is of one of the following forms (we keep the notation of words from \cite{JuhaniIntersection}, although it is not compatible with the notation above; however, we modify integer variables):
\begin{align}
\tag{$\dagger$} &\beta\gamma + \beta(\gamma\beta)^t\left(\delta\left(1 + \gamma\beta + \cdots + (\gamma\beta)^t\right)\right)^*\delta\gamma \\
\tag{$\dagger\dagger$} &\beta\gamma + \beta(\gamma\beta)^{t+m+1}\left(\delta\left(1 + \gamma\beta + \cdots + (\gamma\beta)^t\right)\right)^*\delta(\beta \left(\gamma\beta)^{m}\right)^{-1} 
\end{align}
for some $0 \leq m,t$, where $\delta$ and $\gamma\beta$ are nonempty and $\pref_1(\delta) \neq \pref_1(\gamma\beta)$.

Here 
\begin{align*}
\beta\gamma &= g_0\left(p\,\ee(\zero)\,p^{-1}\right) = h_0\left(q\,\ff(\zero)\,q^{-1}\right) \\
\delta &= g\left(\ee(\one)\right)  = h\left(\ff(\one)\right)\\
\gamma\beta &= g\left(\ee(\zero)\right) = h\left(\ff(\zero)\right) 
\end{align*}
and 
\[
\alpha_gg(p) = \alpha_hh(q) = \left\{
\begin{array}{ll}
\beta(\gamma\beta)^t & \quad\text{for ($\dagger$)} \\[1em]
\beta(\gamma\beta)^{t+m+1} & \quad \text{for ($\dagger\dagger$)} \,.
\end{array}
\right.
\]
The possibility $(\dagger\dagger)$ corresponds to the situation when $f'\ff(\zero)^m$ is a suffix of $\ff(\one)$, where $f'$ is a suffix of $\ff (\zero)$ such that 
$q = f'\ff(\zero)^{m+t+1}$ (and $\beta = \alpha_hh(f')$). In other words, the difference between ($\dagger$) and ($\dagger\dagger$) is whether $\ff(\one)$ contributes to the eventual occurrence of $q$ as a suffix of $\ff(\one\zero^{t+1})$.
\medskip

We finally illustrate the theory by several examples. The first two are from \cite{JuhaniIntersection}.
\begin{example}
	\begin{align*}
	g_0:  \zero &\mapsto \aaa & \one &\mapsto \aaa^m\bbb & \alpha_g &= \alpha = \aaa^m & g:  \zero &\mapsto \aaa & \one &\mapsto \bbb\aaa^m\\
	h_0: \zero &\mapsto \aaa & \one &\mapsto \bbb\aaa^m & \alpha_h &= \emp & h:  \zero &\mapsto \aaa & \one &\mapsto \bbb\aaa^m\ \\
	\ee: \zero &\mapsto \zero & \one &\mapsto \one & p&=\emp & & & & \\
	\ff: \zero &\mapsto \zero & \one &\mapsto \one & q&= \zero^m & t&=m  & & \\[-2em]
	\end{align*}
	\begin{align*}
	T & = \left(\zero +\left(\one + \one\zero + \cdots + \one\zero^{m-1}\right)^*\one\zero^m\right)^* \\
	I & = \aaa +\left(\aaa^m\bbb + \aaa^m\bbb\aaa + \cdots + \aaa^m\bbb\aaa^{m-1}\right)^*\aaa^m\bbb\aaa^m
	\end{align*}
\end{example}	

\begin{example}
	\begin{align*}
	g_0:  \zero &\mapsto \aaa\bbb\aaa & \one &\mapsto \aaa\aaa\bbb & \alpha_g &= \alpha = \aaa & g:  \zero &\mapsto \bbb\aaa\aaa & \one &\mapsto \aaa\bbb\aaa\\
	h_0: \zero &\mapsto \aaa & \one &\mapsto \bbb\aaa\aaa\bbb\aaa & \alpha_h &= \emp & h:  \zero &\mapsto \aaa  & \one &\mapsto \bbb\aaa\aaa\bbb\aaa  \\
	\ee: \zero &\mapsto \zero\zero & \one &\mapsto \one\one & p&=\emp & & & & \\
	\ff: \zero &\mapsto \one\zero & \one &\mapsto \zero\one & q&=\zero & t&=1 & & \\[-2em]
	\end{align*}
	\begin{align*} 
	T & = \left(\zero + \one^+\zero\right)^* \\
	I & = (\aaa\bbb\aaa\aaa\bbb\aaa +\left(\aaa\aaa\bbb\aaa\aaa\bbb\right)^+\aaa\bbb\aaa\aaa\bbb\aaa)^* = \left(\aaa(\aaa\bbb\aaa\aaa\bbb\aaa)^*\bbb\aaa\aaa\bbb\aaa\right)^*
	\end{align*}
\end{example}	

 The noteworthy property of the following example is that $\ff(\zero)$ is a suffix of $\ff(\one)$. The example therefore illustrates the possibility ($\dagger\dagger$) above. 

\begin{example}
	\begin{align*}
	g_0:  \zero &\mapsto \aaa\aaa & \one &\mapsto \aaa^6\bbb & \alpha_g &= \alpha = \aaa^6 & g:  \zero &\mapsto \aaa\aaa & \one &\mapsto \bbb\aaa^6\\
	h_0:  \zero &\mapsto \aaa & \one &\mapsto \bbb\aaa^4 & \alpha_h &= \emp & h:  \zero &\mapsto \aaa  & \one &\mapsto \bbb\aaa^4  \\
	\ee: \zero &\mapsto \zero & \one &\mapsto \one & p&=\emp & & & & \\
	\ff: \zero &\mapsto \zero\zero & \one &\mapsto \one\zero\zero & q&=\zero^6 & t&=2 & & \\[-2em]
	\end{align*}
	\begin{align*} 
	T & = \left(\zero + (\one + (\one\zero)^*\one\zero\zero\right)^* \\
	I & = (\aaa\aaa +\left(\aaa^6\bbb + \aaa^6\bbb\aaa\aaa\right)^*\aaa^6\bbb\aaa\aaa\aaa\aaa)^*
	\end{align*}
\end{example}	

\begin{example}\label{extab}
	Finally, Table \ref{tab} lists various situations in which the intersection is generated by at most one word. 
	Interesting is the last line where all three blocks exist, yet the intersection contains the empty word only. Note that $(p,q)$ is not a suffix of $(pe(\tau),qf(\tau))$ for any nonempty $\tau$ in that case.	
	\begin{table}[ht]
		\bgroup
		\def\arraystretch{1.5}
		\setlength\tabcolsep{.5em}
		\begin{tabular}{c c | c c | c | c | c | c | c | c}
			$g_0(\zero)$ 				 &$g_0(\one)$ 		 & $h_0(\zero)$ 			& $h_0(\one)$  & $\alpha$	& $(p,q)$ 		& $(\ee(\zero),\ff(\zero))$ 		& $$(\ee(\one),\ff(\one))$$ 	& $I$ \\ 
			\noalign{\hrule height 1pt}
			\aaa\aaa\bbb\bbb 	 &\aaa\bbb 		 & \aaa\bbb\aaa 	& \bbb\aaa\bbb &\aaa 	& $(\one,\zero)$ & $\times$	& $(\one\one\one,\one\zero)$ & $\aaa\bbb\aaa\bbb\aaa\bbb^*$ \\  
			\hline
			\aaa\aaa		 	 &\aaa\bbb 		 & \aaa\bbb\aaa 	& \bbb\aaa 		&\aaa	 		& $(\one,\zero)$		& $\times$ 			& $(\one,\one)$ 		& $\{\emp\}$  \\ 
			\hline
			\aaa\aaa\bbb\bbb 	 &\aaa\bbb 		 & \aaa\bbb\aaa 	& \bbb\aaa\bbb\bbb 	&\aaa	& $(\one,\zero)$		& $\times $ 			& $\times$ 		& $\{\emp\}$ \\  
			\hline
			\aaa\aaa\bbb 		 &\aaa\bbb\aaa	 & \aaa\bbb\aaa 	& \bbb\aaa\aaa 			&\aaa	& $\times$		& $(\zero,\zero)$ 			& $(\one,\one)$ 		& $\aaa\bbb\aaa^*$  \\ 
			\hline
			\aaa\aaa\bbb	 	 &\aaa\bbb\bbb	 & \aaa\bbb\aaa 	& \bbb\bbb\aaa 		&\aaa	& $\times$		& $(\zero,\zero)$ 			& $(\one,\one)$ 		& $\{\emp\}$  \\ 
			\hline
			\aaa\aaa\bbb\bbb 	 &\aaa\bbb 		 & \aaa\bbb\aaa\aaa	& \bbb\bbb 		&\aaa			& $\times$ 		& $\times$ 			& $\times$  	& $\aaa\bbb\aaa\aaa\bbb\bbb^*$  \\ 
			\hline
			\aaa\aaa\bbb 		 &\aaa\bbb\bbb	 & \aaa\aaa 		& \bbb\bbb 		&\aaa		& $\times$ 		& $\times$ 			& $\times$  	& $\{\emp\}$  \\ 
			\hline
			\aaa\aaa\bbb 		 &\aaa\bbb\bbb	 & \aaa\aaa\bbb 	& \bbb\bbb\aaa 		&\aaa		& $\times$ 		& $\times$ 			& $(\one,\one)$  	& $\aaa\aaa\bbb^*$  \\ 
			\hline
			\aaa\aaa\bbb 		 &\aaa\bbb\bbb	 & \aaa\bbb\aaa 	& \bbb\aaa\bbb 		&\aaa		& $\times$ 		& $(\zero,\zero)$ 			& $\times$  	& $\{\emp\}$  \\ 
			\hline
			\aaa\bbb\aaa\aaa\bbb &\aaa\bbb\aaa\bbb\aaa\bbb	& \aaa	& \bbb\aaa 		&\aaa\bbb\aaa		& $(\emp,\zero\one)$ 	& $(\zero,\zero\one\one)$ 			& $(\one,\one\one\one)$  	& $\{\emp\}$ \\
			\noalign{\hrule height 1pt} \\[-1em]
		\end{tabular}
		\caption{}\label{tab}
		\egroup
	\end{table}	
	
\end{example}

\section*{Acknowledgments}

The authors acknowledge support by the Czech Science Foundation grant GA\v CR 20-20621S. 

\bibliographystyle{plain}
\bibliography{intersection}	

\begin{thebibliography}{1}

\bibitem{COWhandbook}
Christian Choffrut and Juhani Karhum\"{a}ki.
\newblock Handbook of formal languages, vol. 1.
\newblock chapter Combinatorics of Words, pages 329--438. Springer-Verlag,
  Berlin, Heidelberg, 1997.

\bibitem{IntersectionRevisited}
{\v{S}}t{\v{e}}p{\'a}n Holub.
\newblock Binary intersection revisited.
\newblock In Robert Merca{\c{s}} and Daniel Reidenbach, editors, {\em
  Combinatorics on Words}, pages 217--225, Cham, 2019. Springer International
  Publishing.

\bibitem{JuhaniIntersection}
Juhani Karhum{\"a}ki.
\newblock A note on intersections of free submonoids of a free monoid.
\newblock {\em Semigroup Forum}, 29(1):183--205, Dec 1984.

\bibitem{PaulsonNW-FAC19}
Lawrence~C. Paulson, Tobias Nipkow, and Makarius Wenzel.
\newblock From {LCF} to {Isabelle/HOL}.
\newblock {\em Formal Aspects of Computing}, 31:675--698, 2019.

\bibitem{formalcow-binaryintersection}
Štěpán~Starosta Štěpán Holub.
\newblock Combinatorics on {W}ords {F}ormalized - {B}inary {I}ntersection
  {F}ormalized.
\newblock
  \url{https://gitlab.com/formalcow/combinatorics-on-words-formalized/-/tree/Binary-Intersection-Formalized},
  June 2020.

\end{thebibliography}
	
\end{document}